\documentclass[preprint,10pt]{elsarticle}

\usepackage{amsmath}
\usepackage{color}
\usepackage{bm}
\usepackage{amsfonts}
\usepackage{amsthm}
\usepackage{algorithmic}
\usepackage{amssymb}
\usepackage{tikz}
\usepackage{url}
\usepackage[justification=raggedright]{caption}
\usepackage{graphicx}
\usepackage{subcaption}
\usepackage{microtype}
\PassOptionsToPackage{hyphens}{url}
\usepackage{hyperref}
\hypersetup{colorlinks=true, citecolor=blue, urlcolor=blue, linkcolor=blue}
\newtheorem{theorem}{Theorem}
\newtheorem{proposition}{Proposition}

\newtheorem{lemma}{Lemma}

\DeclareMathOperator{\tr}{tr}

\def\R{{\mathbb R}}  
\newcommand{\C}{\mathbb{C}} 

\newcommand{\NE}{\bm{\mathrm{H}}}
\newcommand{\FINGER}{\widehat{\bm{\mathrm{H}}}}
\newcommand{\TY}{\bm{\mathrm{T}}}
\newcommand{\MT}{\widehat{\bm{\mathrm{T}}}}
\newcommand{\RP}{\bm{\mathrm{R}}}

\journal{}

\begin{document}

\begin{frontmatter}


\title{Fast computation of von Neumann entropy for large-scale graphs via quadratic approximations}



\author{Hayoung Choi\corref{cor1}}
\ead{hchoi@shanghaitech.edu.cn}
\cortext[cor1]{Corresponding author.}

\author{Jinglian He\corref{cor2}}
\ead{hejl1@shanghaitech.edu.cn}

\author{Hang Hu}
\ead{huhang@shanghaitech.edu.cn}

\author{Yuanming Shi}
\ead{shiym@shanghaitech.edu.cn}

\address{School of Information Science and Technology, 
    ShanghaiTech University}

\begin{abstract}
The von Neumann graph entropy (VNGE) can be used as a measure of graph complexity, which can be the measure of information divergence and distance between graphs. However, computing VNGE is extensively demanding for a large-scale graph. 
 We propose novel quadratic approximations for fast computing VNGE. Various inequalities for error between the quadratic approximations and the exact VNGE are found.
 Our methods reduce the cubic complexity of VNGE to linear complexity. 
Computational simulations on random graph models and various real network datasets demonstrate superior performance.
\end{abstract}


\end{frontmatter}


\section{Introduction}
\label{sec:intro}

Graph is one of the most common representations of complex data.
It has the sophisticated capability of representing and summarizing irregular structural features. 
Today, graph-based learning has become an emerging and promising method with numerous applications in many various fields.
With the increasing quantity of graph data, it is crucial to find a way to manage them effectively. Graph similarity as a typical way of presenting the relationship between graphs have been vastly applied \cite{WICKER20132331,DAS2017305,8445612,ABDOLLAHI2015401,2018arXiv180700252C}. 
For example, Sadreazami et al. proposed an intrusion detection methodology based on learning graph similarity with a graph Laplacian matrix \cite{8027126}; also, Yanardag and Vishwanathan gave a general framework to smooth graph kernels based on graph similarity \cite{Yanardag}.
However, all the above approaches relied on presumed models, and thus limited their ability of being applied on comprehending the general concept of divergences and distances between graphs. 

Meanwhile, graph entropy \cite{PhysRevE.80.045102} which is a model-free approach, has been actively used as a way to quantify the structural complexity of a graph. 
By regarding the eigenvalues of the normalized combinatorial Laplacian of a graph as a probability distribution, we can obtain its Shannon entropy.
By giving the density matrix of a graph, which is the representation of a graph in a quantum mechanical state, we can calculate its von Neumann entropy for a graph \cite{PhysRevE.83.036109}. 
Bai et al. proposed an algorithm to solve the depth-based complexity characterization of the graph using von Neumann entropy \cite{6460767}. Liu et al. gave a method of detecting a bifurcation network event based on von Neumann entropy \cite{8461400}. Other applications include graph clustering \cite{6460766}, network analysis \cite{8434737}, and structural reduction of multiplex networks \cite{articleehnail}. 
However, the cost of computing the exact value is highly expensive for a large-scale graph.

\subsection{Von Neumann entropy}
The von Neumann entropy, which was introduced by John von Neumann, is the extension of classical entropy concepts to the field of quantum mechanics \cite{hall2013quantum}. 
He introduced the notion of the density matrix, 
which facilitated the extension of the tools of classical statistical mechanics to the quantum domain in order to develop a theory of quantum measurements. 

Denote the trace of a square matrix $\bm{A}$ as $\tr{\bm{A}}$.
A \emph{density matrix} is a Hermitian positive semidefinite with unite trace.
The density matrix $\bm{\rho}$ is a matrix that describes the statistical state of a system in quantum mechanics.
 The density matrix is especially useful for dealing with mixed states, which consist of a statistical ensemble of several different quantum systems. 

The {\emph{von Neumann entropy}} of a density matrix $\bm{\rho}$, denoted by $\NE (\bm{\rho})$, is defined as
\begin{equation*}
\NE (\bm{\rho}) 
 =   - \tr (\bm{\bm{\rho}}\ln{\bm{\bm{\rho}}}) 
 = - \sum_{i=1}^n \lambda_i \ln \lambda_i,
\end{equation*}
where $\lambda_1,\ldots, \lambda_n$ are eigenvalues of $\bm{\bm{\rho}}$. 
It is conventional to define $0\ln{0} = 0$.
This definition is a proper extension of both the Gibbs entropy and the Shannon entropy to the quantum case.

\subsection{Von Neumann graph entropy}
	
	In this article we consider only undirected simple graphs with non-negative edge weights. 
	Let $G=(V,E,\bm{W})$ denote a graph with the set of vertices $V$ and the set of edges $E$, and the weight matrix $\bm{W}$.
	The \emph{combinatorial graph Laplacian matrix} of $G$ is defined as $\bm{L}(G)=\bm{S}-\bm{W}$, where $\bm{S}$ 
	is a diagonal matrix and its diagonal entry $s_i = \sum_{j=1}^n \bm{W}_{ij}$.
	The {\emph{density matrix}} of a graph $G$ is defined as
		$$\bm{\bm{\rho}}_G =  \frac{1}{\tr{(\bm{L}(G))}} \bm{L}(G),$$
	where $\frac{1}{\tr{(\bm{L}(G))}}$ is a trace normalization factor. 
Note that $\bm{\rho}_G$ is a positive semidefinite matrix with unite trace.
	The {\emph{von Neumann entropy}} for a graph $G$, denoted by $\NE(G)$, 
	is defined as
	\begin{equation*}
		\NE (G) :=  \NE(\bm{\rho}_G),
	\end{equation*}
	where $\bm{\rho}_G$ is the density matrix of $G$. 
	It is also called \emph{von Neumann graph entropy} (VNGE). 
	Computing von Neumann graph entropy requires the entire eigenspectrum $\{ \lambda_i \}_{i=1}^n$ of $\bm{\rho}_G$.
	This calculation can be done with time complexity $\mathcal{O}(n^3)$ \cite{DBLP:books/daglib/0086372}, making it computationally impractical for large-scale graphs.

For example, the von Neumann graph entropy have been proved to be an feasible approach in the computation of Jensen-Shannon distance between any two graphs from a graph sequence \cite{articleehnail}.
However, in the process of machine learning and data mining tasks, a sequence of large-scale graphs will be involved. Therefore, it is of great significance to find an efficient method to compute the von Neumann entropy of large-scale graphs faster than the previous $\mathcal{O}(n^3)$ approach. More details about the application of Jensen-Shannon distance will be shown in Section \ref{sec:applicatoins}.
	To tackle this challenge about computational inefficiency, Chen et al. \cite{2018arXiv180511769C} proposed a fast algorithm for computing von Neumann graph entropy, which uses a quadratic polynomial to approximate the term $-\lambda_i \ln \lambda_i$ rather than extracting the eigenspectrum. 	
	It was shown that the proposed approximation is more efficient than the exact algorithm based on the singular value decomposition. 
Although
it is true that our work was inspired by \cite{2018arXiv180511769C},
the prior work needs to calculate the largest eigenvalue for the approximation.
Our proposed methods does not need the largest eigenvalue to approximate VNGE, so the computational cost is slightly better than the prior work.
Moreover, our proposed methods have superior performances in random graphs as well as real datasets with linear complexity.

\section{Quadratic Approximations}
\label{sec:main}
For a Hermitian matrix $\bm{A} \in \C^{n\times n}$ it is true that $\tr{ f(\bm{A}) } =\sum_{i} f(\lambda_i)$ where $\{ \lambda_i \}_{i=1}^n$ is the eigenspectrum of $\bm{A}$. 
Since $\NE (\bm{\bm{\rho}}) =   - \tr (\bm{\bm{\rho}}\ln{\bm{\bm{\rho}}})$, one natural approach to approximate the von Neumann entropy of a density matrix is to use a Taylor series expansion to approximate the logarithm of a matrix. It is required to calculate $\tr{( \bm{\bm{\rho}}^j )}$, $j \leq N$ for some positive integer $N$.
Indeed, $\ln(\bm{I}_n - \bm{A}) = - \sum_{j=1}^{\infty} \bm{A}^j / j$ for a Hermitian matrix $\bm{A}$ whose eigenvalues are all in the interval $(-1,1)$.
Assuming that all eigenvalues of a density matrix $\bm{\bm{\rho}}$ are nonzeros, we have 
$$\NE(\bm{\bm{\rho}}) = -\ln{\lambda_{\max}} + \sum_{j=1}^{N} \frac{1}{j}   \tr{ \big(    \bm{\rho}(\bm{I}_n - (\lambda_{\max})^{-1} \bm{\rho})^j  \big)  } ,  $$
where $\lambda_{\max}$ is the maximum eigenvalue of $\bm{\rho}$.
We refer to \cite{2018arXiv180101072K} for more details. 
However, the computational complexity is $\mathcal{O}(n^3)$, so it is impractical as $n$ grows. 

In this article, we propose quadratic approximations to approximate the von Neumann entropy for large-scale graphs. It is noted that only $\tr{  (\bm{\bm{\rho}}^2 )  }$ is needed to compute them.
We consider various quadratic polynomials $f_{app} (x)=c_2x^2+c_1x+c_0 $ to approximate $f(x)=-x\ln{x}$ on $(0,1]$ ($f(0)=0$). Then 
$
\NE (\bm{\bm{\rho}}) 
 =   \tr (f( \bm{\rho}) ) \approx  \tr (f_{app}(\bm{\rho}) )
 = \sum_{i=1}^n  f_{app}(\lambda_i)=c_2 \tr{(\bm{\bm{\rho}}^2)} + c_1 + c_0n.
$
Such approximations are required to be considered on the only values in $[0,1]$ such that their sum is 1, which are eigenvalues of a given density matrix.

Recall that $\tr{(\bm{\rho}^2)}$ is called {\emph{purity}} of $\bm{\rho}$ in quantum information theory. 
The purity gives information on how much a state is mixed. 
For a given graph $G$, the purity of $\bm{\bm{\rho}}_G$ can be computed efficiently due to the sparsity of the  $\bm{\bm{\rho}}_G$ as follows.
\begin{lemma}\label{lemma:quadratic_trace1}
For a graph $G=(V,E,\bm{W}) \in \mathcal{G}$,
\begin{equation*}
\tr{( \bm{\bm{\rho}}_G^2 )}= \frac{1}{(\tr{( \bm{L})})^2 } \Bigg(   \sum_{i\in V} \bm{S}_{ii}^2 + 2\sum_{(i,j)\in E} {\bm{L}_{ij}^2}     \Bigg).
\end{equation*}
\end{lemma}
\begin{proof}
Since $\bm{L}$, $\bm{\bm{\rho}}_G$ are symmetric, it follows that 
\begin{align*}
\tr{( \bm{\bm{\rho}}_G^2 )} 
&= ||\bm{\bm{\rho}}_G ||_F^2 
=  \frac{1}{(\tr{\bm{L}})^2}  ||\bm{L} ||_F^2 \\
&=  \frac{1}{(\tr{\bm{L}})^2}  \sum_{1\leq i,j\leq n} \bm{L}_{ij}^2 \\
&=   \frac{1}{(\tr{\bm{L}})^2}  \Bigg(  \sum_{1\leq i\leq n} \bm{L}_{ii}^2 + 2\sum_{1\leq i<j\leq n} \bm{L}_{ij}^2 \Bigg).
\end{align*}
\end{proof}
It is trivial that $\tr{( \bm{\bm{\rho}}_G^2 )}$ only depends on the edge weights in $G=(V,E,\bm{W})$, resulting in linear computation complexity $\mathcal{O}(n+m)$, where $|V|=n$ and $|E|=m$.

We denote the maximum eigenvalue of a given density matrix as $\lambda_{\max}$.
When the eigenspectrum of $\bm{\rho}$, $\{ \lambda_i \}_{i=1}^n$, is given, we denote $\NE(\bm{\rho})$ as $\NE(\lambda_1,\ldots, \lambda_n)$. 
We will use them interchangeably.

\begin{proposition}\label{thm:purity}
The following are true.
\begin{itemize}
\item[(1)] $2\lambda_{\max}(1-\lambda_{\max}) \leq \NE(\lambda_{\max},1-\lambda_{\max}) \leq \NE(\lambda_1,\ldots, \lambda_n)$.

\item[(2)] $1/n \leq \tr{( \bm{\bm{\rho}}^2 )} \leq  \lambda_{\max}^2 + (1-\lambda_{\max})^2$.

\item[(3)] $\tr(-\bm{\rho}^2 \ln{\bm{\rho}}) \leq - \tr{( \bm{\rho}^2 )} \ln{| \tr{( \bm{\rho}^2 )} |}$.

\end{itemize}
\end{proposition}
\begin{proof}
Let $\lambda_n \leq \ldots \leq \lambda_1$ be eigenvalues of a given density matrix $\bm{\rho}$.
\begin{itemize}
\item[(1)] Even though this is one of known properties, we provide a proof here. Let $f(x)=-x\ln{x}$. It is true that $f(p_1+p_2) \leq f(p_1) + f(p_2)$ for all $p_1,p_2\geq 0$. Indeed,
\begin{align*}
  f(p_1) + f(p_2)& -f(p_1+p_2)
     = -p_1 \ln{p_1} -p_2 \ln{p_2} +(p_1+p_2) \ln{(p_1+p_2)}\\
     &= p_1(\ln{(p_1+p_2)} - \ln{p_1})  +
     p_2(\ln{(p_1+p_2)}-\ln{p_2}) \geq 0.
\end{align*}
Thus, by induction we have
$$ f(\lambda_2+\cdots + \lambda_n)  \leq f(\lambda_2) + \cdots + f(\lambda_n).$$
Therefore, $f(\lambda_1) +  f(\lambda_2+\cdots + \lambda_n)  \leq f(\lambda_1) + \cdots + f(\lambda_n) = \NE(\bm{\rho})$.

The first inequality holds from the fact that $x-1 \geq \ln{x}$ for all $x>0$.
Thus, 
\begin{align*}
-\lambda_{\max} \ln{\lambda_{\max}} - (1-\lambda_{\max})\ln{(1-\lambda_{\max})} 
& \geq -\lambda_{\max}(\lambda_{\max}-1) -(1-\lambda_{\max})(-\lambda_{\max})\\
& = 2\lambda_{\max}(1-\lambda_{\max}).
\end{align*}

\item[(2)]  Since $\tr( \bm{\rho}^2) \leq \lambda_1^2 +(\lambda_2 +\cdots + \lambda_n)^2$,  clearly $\tr( \bm{\rho}^2) \leq \lambda_1^2 +(1- \lambda_1)^2$. For a proof of the first inequality, see \cite{Jaeger2007}.

\item[(3)] 
By Proposition \ref{thm:purity} (2), clearly $1/n \leq \tr{( \bm{\rho}^2 )} \leq 1$. 
Since $f(x)=-x\ln{x}$ is concave on $[0,1]$, by Jensen's inequality, it follows that
\begin{align*}
f(\tr{( \bm{\rho}^2 )}) 
= f \Bigg( \sum_{i=1}^n \lambda_i^2 \Bigg) 
\geq \sum_{i=1}^n \lambda_i f(\lambda_i)
= \sum_{i=1}^n -\lambda_i^2 \ln{\lambda_i}.
\end{align*}

\end{itemize}
\end{proof}

\begin{figure}[t!]
  \centering
    \includegraphics[width=0.8\textwidth]{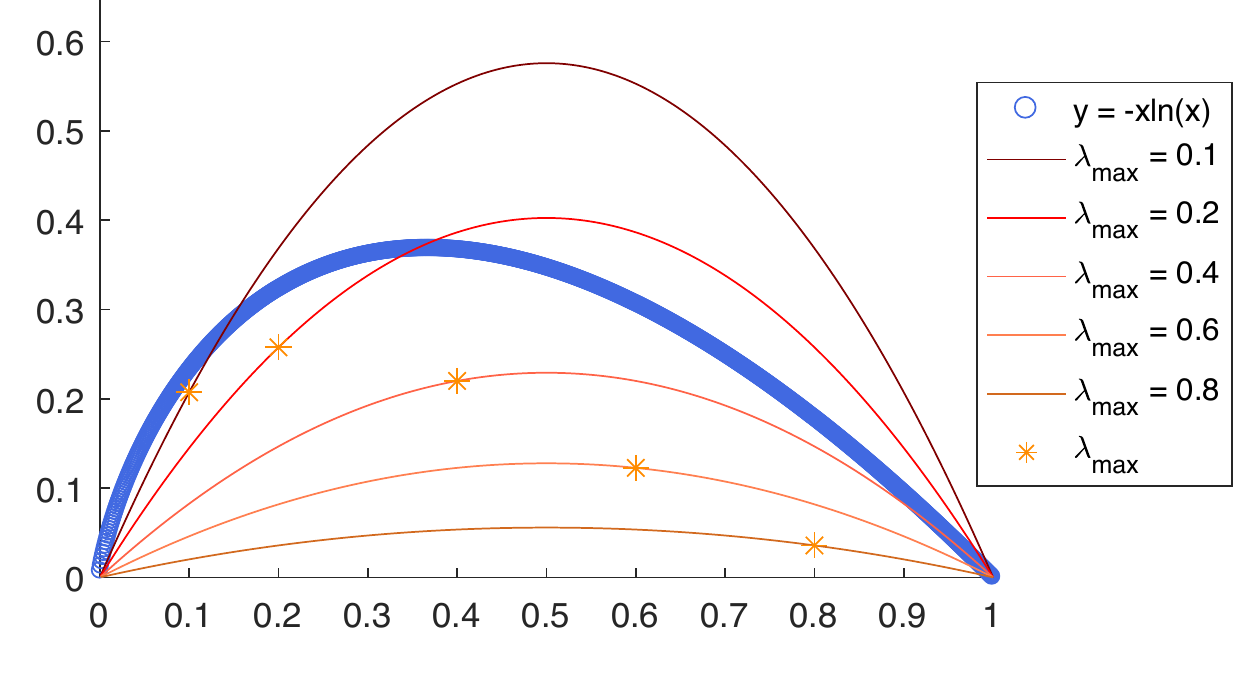}
     \caption{FINGER-$\FINGER$: graphs of the quadratic function in \eqref{eq:finger_app} with different $\lambda_{\max}$ are shown. } \label{fig:FINGER}
\end{figure}

This article considers the following quadratic approximations for von Neumann entropy:
(i) FINGER-$\FINGER$; (ii) Taylor-$\TY$; (iii) Modified Taylor-$\MT$; (iv) Radial Projection-$\RP$.
Note that they all can be computed by the purity of density matrix of a graph. Additionally (i) and (iii) need to compute the maximum eigenvalue as well. 
\subsection{FINGER-\texorpdfstring{$\FINGER$}{Lg}}

Chen et al.\cite{2018arXiv180511769C} proposed 
\emph{a fast incremental von Neumann graph entropy} (FINGER) to reduce the cubic complexity of von Neumann graph entropy to linear complexity.
They considered the following quadratic polynomial to approximate $f(x)=-x\ln{x}$ (see Fig. \ref{fig:FINGER}).
\begin{equation}\label{eq:finger_app}
q(x)=-(\ln{\lambda_{\max}})x(1-x) \quad \text{on }[0,1].
\end{equation}
Then FINGER, denoted by $\FINGER$, is defined as
\begin{align*}
\FINGER(G)
&:= \sum_{i=1}^n q(\lambda_i)
= -(\ln{\lambda_{\max}})(1-  \tr(\bm{\rho}_G^2)),
\end{align*}
where $\lambda_1,\ldots, \lambda_n$ are the eigenvalues of $\bm{\bm{\rho}}_G$.

Note that since all eigenvalues are smaller than or equal to the maximum eigenvalue, 
it is better to deal with the functions $f$ and $q$ on the interval $[0,\lambda_{\max}]$ instead of $[0,1]$.
Now let us show that the approximation FINGER is always smaller than the exact von Neumann entropy. 
\begin{lemma}\label{lemma:FINGER_app}
The following are true.
\begin{itemize}
\item[(1)] $f-q$ is concave on $[ 0, \lambda_{max} ]$.

\item[(2)] $f(x) - q(x) \geq \lambda_{\max} f(x)$ on $[ 0, \lambda_{max} ]$.
\end{itemize}
\end{lemma}
\begin{proof}
 It is trivial for $\lambda_{\max}=1$. Suppose that $\lambda_{\max} \neq 1$.
\begin{itemize}
\item[(1)]
Note that $f(x)< 1/e <1/2$ on all $0 \leq x \leq 1$. So, $- \lambda_{\max} \ln{ \lambda_{\max}} < 1/2$. Then $\lambda_{\max} < \frac{1}{-2\ln{ \lambda_{\max}}}$, implying, 
$-1-2x \ln{ \lambda_{\max}}  < 0$ for all $0 \leq x \leq \lambda_{\max}$. Thus, 
$$(f-q)''(x)= \frac{-1-2x \ln{ \lambda_{\max}}}{x} <0.$$

\item[(2)] When $x=0$ or $x=\lambda_{\max}$, it is trivial.
Since $(1- \lambda_{\max})f(x) -q(x) = x((\lambda_{\max}-1)\ln{x}  + (\ln{\lambda_{\max}}) (1-x))$,
it suffices to show $\phi(x)=(\lambda_{\max}-1)\ln{x}  + \ln{\lambda_{\max}} (1-x) >0$ for all $0 < x < \lambda_{\max}$. By observing the tangent line at $x=1$ for $y=x\ln{x}$, it is easy to check that $\lambda_{\max} -1 < \lambda_{\max} \ln{\lambda_{\max}}$. Then 
$$\frac{\lambda_{\max}-1}{\ln{\lambda_{\max}}} > \lambda_{\max}.$$ Thus, $\phi '(x)=\frac{\lambda_{\max}-1}{x} - \ln{\lambda_{\max}} <0$  for all $0<x< \lambda_{\max}$.  
Since $\phi(\lambda_{\max})=0$ and $\phi '(x) <0$ for all $0<x<\lambda_{\max}$, it follows that $\phi(x) >0$ for all $0<x< \lambda_{\max}$.

\end{itemize}
\end{proof}

\begin{theorem}\label{thm:FINGER_eigenmax}
For any density matrix $\bm{\rho}$, it holds that
\begin{itemize}
    \item[(1)] $\FINGER(\lambda_1,\ldots, \lambda_n) \geq \FINGER(\lambda_{\max},1-\lambda_{\max})$.

    \item[(2)]  $\NE(\bm{\rho}) \geq  \FINGER(\bm{\rho})$ The equality holds if and only if $\lambda_{\max} =1$. 
    
    \item[(3)] $\NE(\bm{\rho}) -  \FINGER(\bm{\rho})  \geq  \lambda_{\max} \NE(\bm{\rho}).$
\end{itemize}
\end{theorem}
\begin{proof}
\begin{itemize}
    \item[(1)] Let $\phi(x)=x(1-x)$. Then it is easy to check that $ \phi(t_1) + \phi(t_2) - \phi(t_1+t_2) \geq 0$ for all $t_1,t_2 \geq 0$.
Thus,  $q(t_1 + t_2) \leq q(t_1) + q(t_2)$. By induction, it holds that
$$ q(t_1 + \cdots + t_n) \leq q(t_1) + \cdots q(t_n),$$
for all $t_i\geq 0$, $i=1,\ldots,n$.
Then it follows that
\begin{align*}
\FINGER(\lambda_1,\ldots, \lambda_n)
& = \sum_{i=1}^n q( \lambda_i) 
=  q( \lambda_1) + \sum_{i=2}^n q( \lambda_i) \\
&\leq q( \lambda_1) + q\Big( \sum_{i=2}^n \lambda_i \Big)\\
& = q( \lambda_{1}) + q(1- \lambda_{1}),
\end{align*}
where $\lambda_n \leq \cdots \leq \lambda_1$ are spectrum of $\bm{\rho}$.

    \item[(2)] Let $\lambda$ be an eigenvalue of $\bm{\rho}$.
Since $\lambda \leq \lambda_{\max}$, it follows that
\begin{align*}
f(\lambda) - q(\lambda)
&  = -\lambda \ln{\lambda} + (\ln{\lambda_{max}})\lambda(1-\lambda)\\
& \geq -\lambda \ln{\lambda} + (\ln{\lambda})\lambda(1-\lambda)\\
& = -\lambda^2 \ln{\lambda} \geq 0.
\end{align*}
Note that $0 = \lambda_n = \cdots = \lambda_2$ and $\lambda_1=1$ if and only if     $\lambda_i^2 \ln{\lambda_i} = 0$ for all $1\leq i \leq n$. 
\item[(3)] By Lemma \ref{lemma:FINGER_app} (2) we can find the  error bound for FINGER.

\end{itemize}
\end{proof}

For more information about FINGER, see
\cite{2018arXiv180511769C}.

\subsection{Taylor-\texorpdfstring{$\TY$}{Lg}}

Since the sum of all eigenvalues of density matrix is $1$, the average of them is $\frac{1}{n}$. 
As $n$ gets bigger, the average gets closer to $0$. Thus, for a large-scale $n\times n$ density matrix, many of its eigenvalues must be on $[0, \frac{1}{n}]$.
So, it is reasonable to use Taylor series for $f$ at $x=\frac{1}{n}$ instead of $x=0$. In fact, since $f'(0)$ does not exist, there does not exist Taylor series for $f$ at $x=0$. 
\begin{lemma}\label{lemma:lambda_max}
Let $ \lambda_{\max} $ be the maximum eigenvalue of a density matrix $\bm{\rho}$.
Then,
$\frac{1}{n} \leq \lambda_{\max} \leq 1$.
Especially, $ \lambda_{\max} = 1$ if and only if $\NE(\bm{\rho})=0$. Also it is true that $ \lambda_{\max} = \frac{1}{n}$ if and only if $\NE(\bm{\rho})=\ln(n)$.
\end{lemma}
\begin{proof}
Since $\tr{\bm{\rho}} =1$, if $\lambda_{\max}= 1$ then all eigenvalues are 0 except $\lambda_{\max}$.
It is known that $\NE(\bm{\rho}) \leq \ln(n)$ and the equality holds when $\lambda_i = \frac{1}{n}$ for all $i$.\end{proof}

We can propose the quadratic Taylor approximation for $f$ at $x=\frac{1}{n}$ as follows.
\begin{align*}
q(x) 
&= f\bigg(\frac{1}{n}\bigg) + f'\bigg(\frac{1}{n}\bigg)\bigg(x-\frac{1}{n}\bigg) + \frac{f''(\frac{1}{n})}{2!}\bigg(x-\frac{1}{n}\bigg)^2\\
&= -\frac{n}{2} x^2 + (\ln{n}) x - \frac{1}{2n}.
\end{align*}
Using such approximation, Taylor, denoted by $\TY$, is defined as
\begin{align*}
\TY(G)
&= -\frac{n}{2} \tr(\bm{\rho}_G^2) + \ln{n} -\frac{1}{2}.
\end{align*}

As Fig. \ref{fig:graph-Taylor-T} shows, 
the function $q$ is very similar to the function $f$ near $x=\frac{1}{n}$. However,
as the maximum eigenvalue gets closer to 1, the error becomes very large. Note that $f(\lambda_{\max}) - q(\lambda_{\max}) = -\lambda_{\max} \ln{\lambda_{\max}} + \frac{n}{2}\lambda_{\max}^2 - (\ln{n}) \lambda_{\max} - \frac{1}{2n} \longrightarrow \infty$ as $n\longrightarrow \infty$ for $\lambda_{\max}\approx 1$.
Alternatively, this approximation needs to be modified. We use the information about $\lambda_{\max}$ in order to reduce the error.

We assume that the eigenvalues concentrate around the mean $1/n$. Remark that this would be in general true for small-world or relatively small-world graphs, but for example in planar graphs or graphs representing a low dimensional manifold where Weyl's law holds this would not be true. In particular for planar geometry or $2D$ manifolds the smallest eigenvalues would grow linearly, and this rate would most likely hold well around $1/n$. 

\subsection{Modified Taylor-\texorpdfstring{$\MT$}{Lg}}

\begin{figure}[t!]
  \centering
    \includegraphics[width=0.8\textwidth]{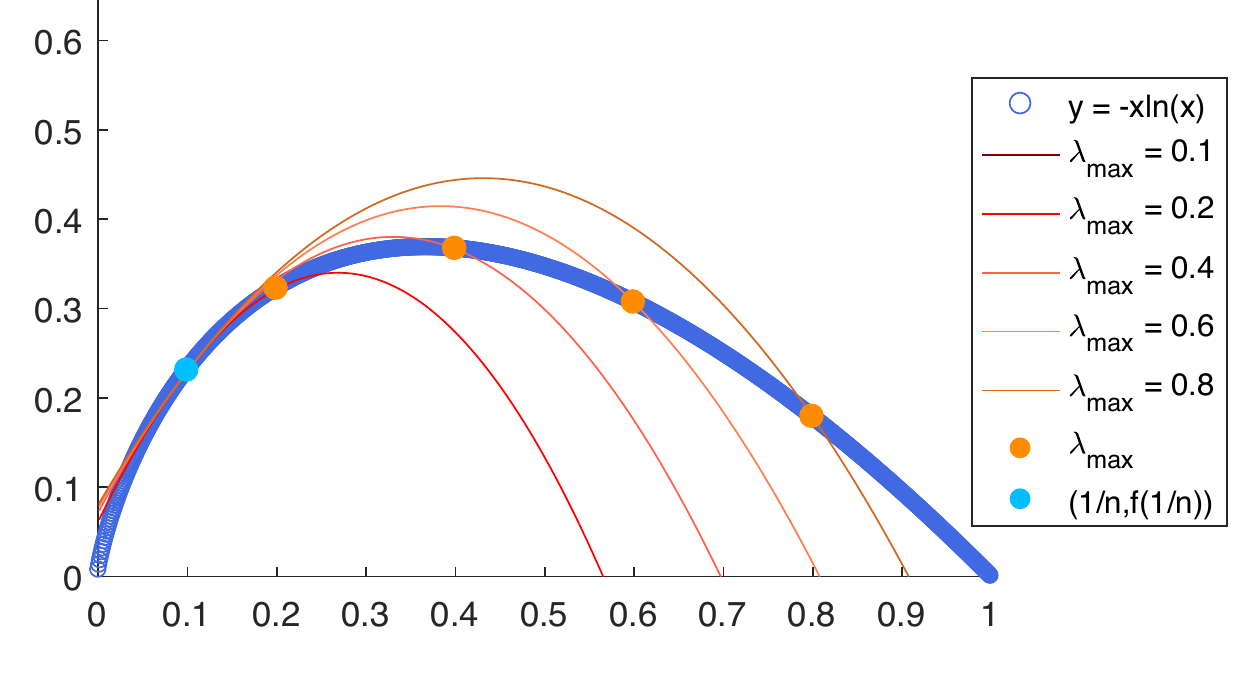}
     \caption{Modified Taylor-$\MT$: graphs of the quadratic function in \eqref{eq:modified_Taylor_app} with different $\lambda_{\max}$ are shown.}
     \label{fig:graph-Taylor-T}
\end{figure}

Consider the quadratic approximation, $q$, to approximate $f(x)=-x\ln{x}$ such that it holds
 \begin{equation*}
 q\bigg(\frac{1}{n}\bigg) = f\bigg(\frac{1}{n}\bigg),~ q'\bigg(\frac{1}{n}\bigg) = f'\bigg(\frac{1}{n}\bigg),~ q(\lambda_{\max}) = f(\lambda_{\max}),
 \end{equation*}
 assuming that $\lambda_{\max} \neq \frac{1}{n}$,
we have
\begin{equation}\label{eq:modified_Taylor_app}
q(x) 
= \sigma x^2 + \bigg(\ln{n}-1-\frac{2\sigma}{n} \bigg)x + \frac{\sigma+n}{n^2},
\end{equation}
where 
\begin{equation*}
\sigma= \frac{ -n\lambda_{\max}\ln{(n\lambda_{\max})} + n\lambda_{\max} -1}{n(\lambda_{\max}-\frac{1}{n})^2}.
\end{equation*}

Using such approximation, the Modified Taylor, denoted by $\MT$, is defined as
\begin{align*}
\MT(G)
&= \sigma \bigg( \tr(\bm{\rho}_G^2) -\frac{1}{n} \bigg) +\ln{n}.
\end{align*}

\begin{lemma}\label{lemma:modifiedT}
The following are true.
\begin{itemize}
\item[(1)] $q$ is concave on $[0, \lambda_{\max}]$.

\item[(2)] $\frac{1}{n} \leq -\frac{1}{2 \sigma} \leq  \lambda_{\max} $. 
\end{itemize}
\end{lemma}
\begin{proof}
\begin{itemize}
\item[(1)] 
Since $q$ is a quadratic polynomial in $x$, it suffices to show $\sigma <0$.  Let $\phi(t) = -t \ln{t} +t -1$. Since $\phi(1)=0$ and $\phi'(t) <0$ for all $t > 1$, it is true that $\phi(t) < 0$ for all $t >1$.  By Lemma \ref{lemma:lambda_max}, $n \lambda_{\max} > 1$, so  $\phi(n \lambda_{\max}) = -n \lambda_{\max} \ln{(n \lambda_{\max})} +n \lambda_{\max} -1 <0 $.

\item[(2)] Let $t=n \lambda_{\max}$. By Lemma \ref{lemma:lambda_max}, $t > 1$. Since 
\begin{equation*}
\frac{1}{-2\sigma} -\frac{1}{n} = \frac{1}{n} \Bigg( \frac{(t-1)^2}{2(t \ln{t} -t +1)}  -1 \Bigg),
\end{equation*}
It suffices only if we show $(t-1)^2 > 2(t \ln{t} -t +1)$. 
Recall that $t \ln{t} -t +1>0$ for all $t>1$.
Let $\phi(t) = t^2 -2t\ln{t} -1 $.
Since $\ln{t} < t-1$ for all $t$, $\phi'(t) = 2t-2\ln{t} -2 = 2((t-1)-\ln{t}) >0$ for all $t\geq 1$. From the fact $\phi(1)=0$, clearly, $\phi(t) > 0$ for all $t>1$.
So, $\frac{1}{n} \leq -\frac{1}{2 \sigma} $.
Taking $t = n\lambda_{\max} $, we have 
\begin{align*}
\lambda_{\max}  - \frac{1}{-2\sigma} 
&= \frac{t}{n} - \frac{(t-1)^2}{2n(t \ln{t} -t +1)} \\
&= \frac{2t(t \ln{t} -t+1) -(t-1)^2}{2n(t\ln{t} -t +1)} >0.
\end{align*}
\end{itemize}
\end{proof}

\begin{theorem}\label{thm:Modified_Taylor_error}
For any density matrix $\bm{\rho}$, it holds that
\begin{equation*}
\MT(\bm{\rho}) \geq \NE(\bm{\rho}) .
\end{equation*}
\end{theorem}
\begin{proof}
Let $h(x)=q(x) - f(x)$.
Clearly, $h(1/n)= h(\lambda_{\max})=0$ and $h'(1/n)=0$.
We show that $h(x) \geq 0$ for three different intervals: (i) $[0,\frac{1}{n}]$, (ii) $[\frac{1}{n},-\frac{1}{2\sigma}]$, and (iii) $[-\frac{1}{2\sigma},\lambda_{\max}]$.
(i) Since $h'' = 2\sigma + \frac{1}{x}$, Lemma \ref{lemma:modifiedT} (2) implies that $h$ is convex on $[0,-\frac{1}{2\sigma}]$. Since $h(1/n)=h'(1/n)=0$, $h(x) \geq 0$ on $[0,\frac{1}{n}]$.
(ii) In the similar way, it holds that $h(x) \geq 0$ on $[\frac{1}{n},-\frac{1}{2\sigma}]$.
(iii) Since $h$ is concave on $[-\frac{1}{2\sigma},\lambda_{\max}]$,
by the definition of concavity, it follows that 
$$ h\Big( -\frac{1}{2\sigma} t + \lambda_{\max}(1-t) \Big) \geq h\Big( -\frac{1}{2\sigma} \Big)t + h( \lambda_{\max})(1-t) \geq 0$$
for all $0\leq t\leq 1$.
The last inequality holds from the fact that $h(-\frac{1}{2\sigma}) \geq 0$ and $h( \lambda_{\max}) = 0$. 
Thus, $h(x) \geq 0$ on $[-\frac{1}{2\sigma},\lambda_{\max}]$.
Therefore, by (i), (ii), (iii) it holds that $h(x) \geq 0$ on $[0,\lambda_{\max}]$.
\end{proof}

\subsection{ Radial Projection-\texorpdfstring{$\RP$}{Lg}}

We denote the simplex of positive probability as $\Delta_n$, i.e., 
$$\Delta_n:=\Big\{(\lambda_1,\lambda_2,\ldots,\lambda_n) \Big| \sum_{i=1}^n \lambda_i= 1,~  \lambda_i  \geq 0 \Big\}.$$ 
Also we denote $\mathcal{E}:=(1 / n,\ldots, 1 / n)$. Clearly, $\mathcal{E}\in \Delta_n$.
The {\emph{Shannon entropy}} of $\Lambda \in \Delta_n$ is defined as  
$\mathcal{S}(\Lambda) = -\sum_{i=1}^n \lambda_i \ln{\lambda_i}.$
It is well-known that it holds
$\mathcal{S}(\Lambda) \leq \ln{n}$, with equality if  $\Lambda = \mathcal{E}$.
That is, $\mathcal{E}$ is the only point where the entropy is maximum.


As Fig. \ref{fig:VNE_3D} is shown, 
the simplex of positive probability $\Delta_3$ can be geometrically presented as
a part of plane in $\R^3$. The color stands for the value of Shannon entropy at each point. 
One can see that
as $\Lambda \in \Delta_3$ gets closer to $\mathcal{E}$, the entropy gets bigger and bigger.
Our main observation is that 
if two points on $\Delta_3$ have same (Euclidean) distances from $\mathcal{E}$, then their purity are same. In general it holds for any $n$.

\begin{lemma}\label{lemma:radial1}
Let $\Lambda=(\lambda_1,\lambda_2,\ldots,\lambda_n)\in \Delta_n$. Then the following are true.
\begin{itemize}
\item[(1)]
$||\mathcal{E}-\Lambda||_2^2 = \sum_{i=1}^n \lambda_i^2 -\frac{1}{n}.$

\item[(2)] $\langle \mathcal{E}, \mathcal{E}-\Lambda \rangle =0.$ ($\langle \cdot, \cdot \rangle$ is the usual inner product.)

\end{itemize}
\end{lemma}
\begin{proof}
It is easy to check that
\begin{align*}
||\mathcal{E}-\Lambda||_2^2 
=   \sum_{i=1}^n \bigg( \lambda_i-\frac{1}{n} \bigg )^2 
 = \sum_{i=1}^n \lambda_i^2 -\frac{1}{n}. 
\end{align*}
\end{proof}

\begin{figure}[t!]
  \centering
    \includegraphics[width=0.8\textwidth]{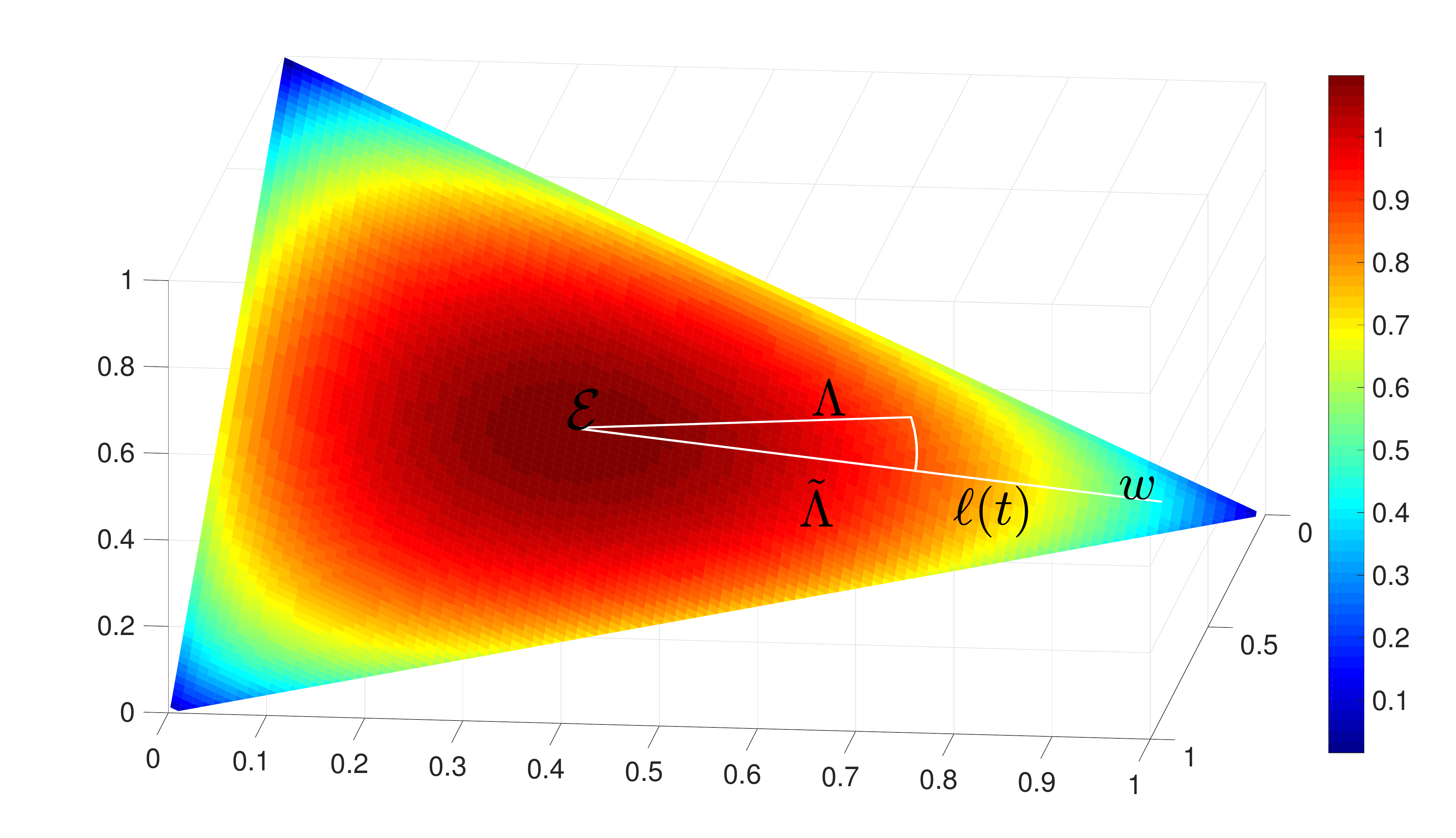}
     \caption{The Shannon entropy at each point $(\lambda_1,\lambda_2,\lambda_3) \in \Delta_3$ with different color levels.}
     \label{fig:VNE_3D}
\end{figure}

\begin{theorem}\label{thm:dis_purity1}
Let  $\bm{\rho}$ and $\tilde{\bm{\rho}}$ be $n\times n$ density matrices. 
Then $|| \bm{\rho} - 1/n \bm{I} ||_F = || \widetilde{\bm{\rho}} - 1/n \bm{I} ||_F$ if and only if their purity are identical.
\end{theorem}
\begin{proof}
Note that $|| \bm{\rho} - 1/n \bm{I} ||_F = || \widetilde{\bm{\rho}} - 1/n \bm{I} ||_F$ if and only if $||\mathcal{E}-\Lambda ||_2 =  ||\mathcal{E}-\widetilde{\Lambda} ||_2$, 
where $\lambda_1,\lambda_2,\ldots,\lambda_n$ and $\tilde{\lambda}_1,\tilde{\lambda}_2,\ldots,\tilde{\lambda}_n$ are eigenvalues of $\bm{\rho}$ and $\widetilde{\bm{\rho}}$, respectively. It is true from Lemma \ref{lemma:radial1}. 
\end{proof}

Lemma \ref{lemma:radial1} states that two points on $\Delta_n$ have the same distance from $\mathcal{E}$ if and only if the distances from the origin are identical.
Then
we can find $\tilde{\Lambda} \in \Delta_n$ whose entropy can be computed much easily such that
\begin{equation*}
||\mathcal{E}-\Lambda||_2 =||\mathcal{E}-\widetilde{\Lambda} ||_2.
\end{equation*}
There are infinitely many directions from $\mathcal{E}$ to find $\tilde{\Lambda} \in \Delta_n$. Among them we
pick  $\bm{w}=(c,\frac{1-c}{n-1},\ldots, \frac{1-c}{n-1}) \in \Delta_n$ with $c=\lambda_{\max}$. 
 
We consider the line segment $\ell(t)= (\bm{w}-\mathcal{E})t+\mathcal{E}$, $0 \leq t \leq 1$, i.e., 
\begin{equation*}
\ell(t)= \bigg(\frac{1-t}{n}+tc,\frac{1-t}{n}+\frac{t(1- c)}{n-1},\ldots, 
\frac{1-t}{n}+\frac{t(1- c)}{n-1} \bigg).
\end{equation*}
Since $\Delta_n$ is convex, $\ell(t) \in \Delta_n$ for all $0 \leq t \leq 1$. 
For each $0\leq t \leq 1$, we have 
\begin{align*}
S(\ell(t))
&= -\bigg(\frac{1-t}{n}+tc\bigg) \ln{\bigg( \frac{1-t}{n}+tc \bigg)}\\
&- (n-1) \bigg(\frac{1-t}{n}+\frac{t(1- c)}{n-1} \bigg)\ln{\bigg(\frac{1-t}{n}+\frac{t(1- c)}{n-1} \bigg)}.
\end{align*}
Lemma \ref{lemma:radial1} implies that 
\begin{align*}
||\mathcal{E}-\ell(t)||_2^2 
= t^2 ||\bm{w} - \mathcal{E} ||_2^2 
=   \frac{(cn-1)^2}{n(n-1)}  t^2.
\end{align*}

We solve the following equation for $0\leq t \leq 1$:
\begin{equation*}
||\mathcal{E}-\Lambda||_2 =||\mathcal{E}-\ell(t)||_2.
\end{equation*}
Then the solution, say $t_0$, is
\begin{equation*}
t_0=   \frac{\sqrt{n(n-1)}}{cn-1}  \sqrt{  \sum_{i=1}^n   \lambda_i^2 -\frac{1}{n} }.
\end{equation*}

Since $||\mathcal{E}-\Lambda||_2 =||\mathcal{E}-\ell(t_0)||_2$, we have that
\begin{align*}
&S(\Lambda) \approx S(\ell(t_0)) = -\bigg(\frac{1-t_0}{n}+t_0c\bigg) \ln{\bigg( \frac{1-t_0}{n}+t_0c \bigg)}\\
&- (n-1) \bigg(\frac{1-t_0}{n}+\frac{t_0(1- c)}{n-1} \bigg)\ln{\bigg(\frac{1-t_0}{n}+\frac{t_0(1- c)}{n-1} \bigg)}.
\end{align*}
In fact, putting $t_0$ into the right side the constant $c$ can be cancelled.
Thus, this approximation does not need the maximum eigenvalue.

Now we propose the quadratic approximation for the von Neumann graph entropy, called {Radial Projection}, denoted by $R$, is defined as

\begin{align}
R(G)
&= -\bigg( \sqrt{\frac{n-1}{n}}\kappa_G+\frac{1}{n} \bigg) \ln{\bigg( \sqrt{\frac{n-1}{n}}\kappa_G+\frac{1}{n} \bigg)}  \nonumber \\
&\quad - (n-1) \bigg(  -\frac{1}{\sqrt{(n-1)n}}\kappa_G+\frac{1}{n} \bigg)\ln{\bigg( -\frac{1}{\sqrt{(n-1)n}}\kappa_G+\frac{1}{n} \bigg)}, \label{eq:R_form}
\end{align}
where $\kappa_G = [\tr(\bm{\rho}_G^2) - \frac{1}{n}  ]^{1/2}$.

\subsection{Weighted mean}
 Denote the weighted mean of $a,b \in \R$ as $a\#_t b$. That is, 
$a\#_t b = ta +(1-t)b. $
In a similar way we consider $\FINGER(\bm{\rho}) \#_t \MT(\bm{\rho})$ and $\FINGER(\bm{\rho})  \#_t \RP(\bm{\rho})$.
By Theorem \ref{thm:FINGER_eigenmax} (2), it is shown that FINGER-$\FINGER$ is always smaller than the exact von Neumann entropy.
On the other hand, by Theorem \ref{thm:Modified_Taylor_error}
Modified Taylor-$\MT$ is always greater than the exact von Neumann entropy.
Even though it is not proved mathematically, Fig. \ref{fig:VNE_3D}, \ref{exact_approx1}, \ref{exact_approx2} show that 
Radial projection-$\RP$ are greater than the exact von Neumann entropy.
The weighted mean of them can be computed to improve the approximations.

We solve a optimization problem to find optimal $t^*$. 
For example, consider $\FINGER \#_t \MT$. 
Given large quantity of real data sets, the approximation of von Neumann entropy using $\FINGER$ and $\MT$ were calculated as the input values $ x^{(i)}_{\widehat{T}}$ and $x^{(i)}_{\widehat{H}}$ while the actual von Neumann entropy value were also calculated as output values $y^{(i)}$ for $i=1,\ldots, N$ ($N$ is the number of data sets). 
Then the optimization problem is given as follows:
$$t^* =   \mathop{\arg\min}_{0\leq t\leq 1} J(t), $$
where the cost function is
$$J(t) = \frac{1}{N} \sum^{N}_{i = 1} \Big(t x^{(i)}_{\widehat{H}} + (1- t) x^{(i)}_{\widehat{T}}-y^{(i)} \Big)^2.$$
We use the gradient descent method to find optimal $t^*$. 
Initially, $t_0=\frac{1}{2}$ is given.
In each step of gradient descent, $t_j$ is updated with the function:
$$t_j = t_{j-1}- \alpha  J'(t_{j-1}),$$
where $j$ denotes iteration times, and $\alpha$ is the step size, which is set to be $10^{-6}$.

One may wander if such \textquotedblleft real data sets" are independent of the test networks employed in the numerical experiments. On the one hand they should, in order to avoid a circular argument; on the other hand, to be effective such \textquotedblleft training" should somehow extract the generalities of a class of graphs. 
In order to eliminate overfitting, we separated our datasets into two parts, the training datasets and the test datasets. After we got the optimal weights $t^*$ using the process stated above with training datasets, we calculate the approximate value $V_1$ of the test datasets. Then we use test datasets to find the optimal $t^*_{test}$ and calculate the approximate value $V_2$ using $t^*_{test}$, the average difference between $V_1$ and $V_2$ are lower than $0.02$, which is acceptable. 
However, it is still questionable to show mathematically that  $t^*$ is optimal.

Using the optimal values $t^*$ solved by the gradient descent method, we call $\FINGER\#_{0.3824} \MT$ and $\FINGER\#_{0.2794} \RP$ as 
\emph{Improved Modified Taylor} and
\emph{Improved Radial Projection}, respectively. 
As anonymous referee suggested, we also consider 
$\omega_1\widehat{\mathbf{H}} + \omega_2 \mathbf{T} + \omega_3  \widehat{\mathbf{T}}  + \omega_ 4 \mathbf{R} + \beta$, where $\omega_i$ are weights such that
$0\leq \omega_i \leq 1$ and $\omega_1+\cdots + \omega_4 =1$, and 
$\beta \in \mathbb{R}$ is a constant shift. In a similar way, the optimal weights can be found.
It is shown that $0.2299\widehat{\mathbf{H}}+0.3099\widehat{\mathbf{T}}+0.4602 \mathbf{R}-0.0073$ has superior performances for real datasets.
We call this as \emph{Mixed Quadratic approximiation}.

\section{Experiments}
\label{sec:experiments}

\subsection{Random graphs}
\begin{figure}[t!]
\begin{center}
\begin{minipage}[b]{.45\linewidth}
  \centering
  \centerline{\includegraphics[width=1\linewidth]{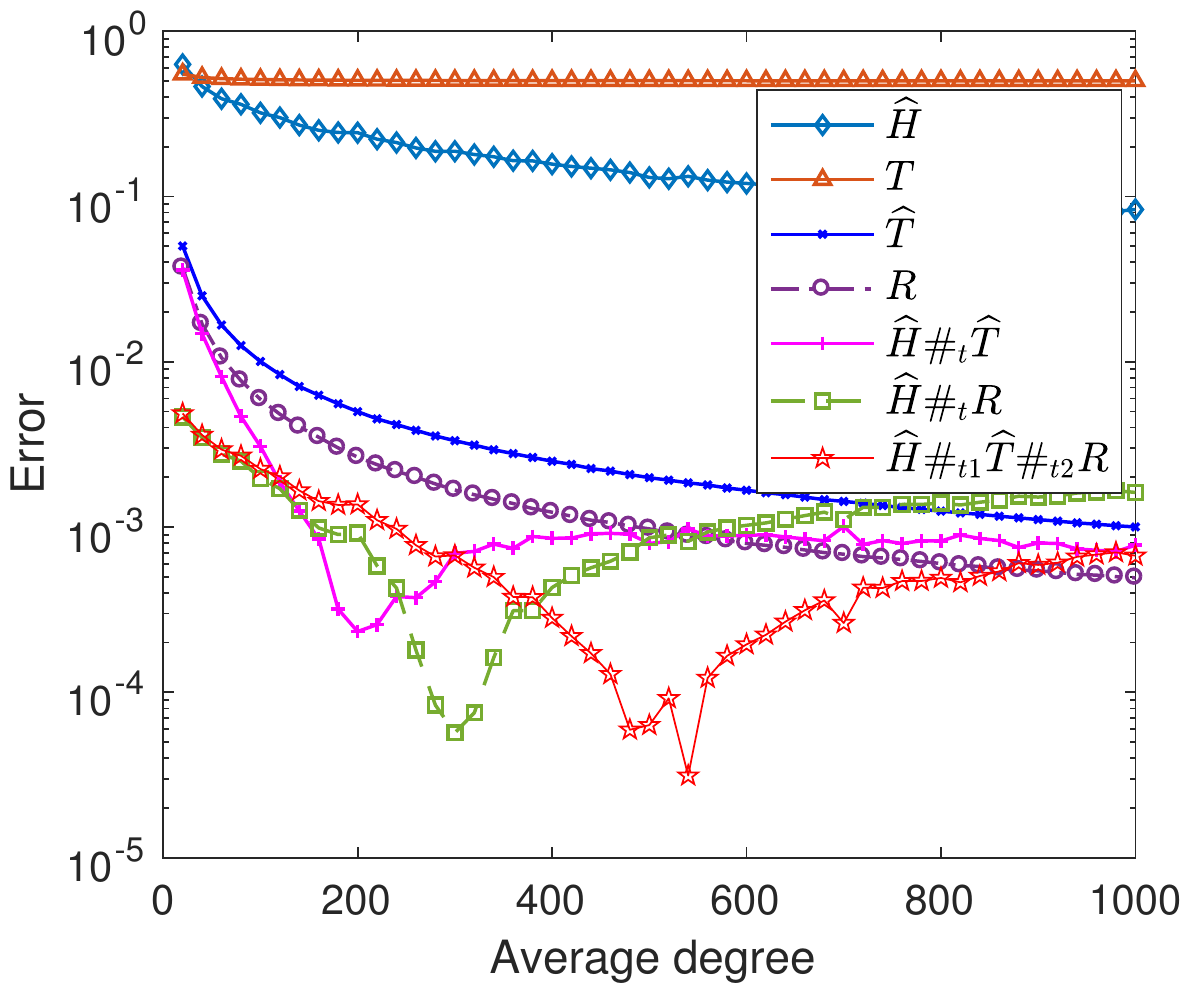}}
  \centerline{(a) ER average degree}\medskip
\end{minipage}
\begin{minipage}[b]{.45\linewidth}
  \centering
  \centerline{\includegraphics[width=1\linewidth]{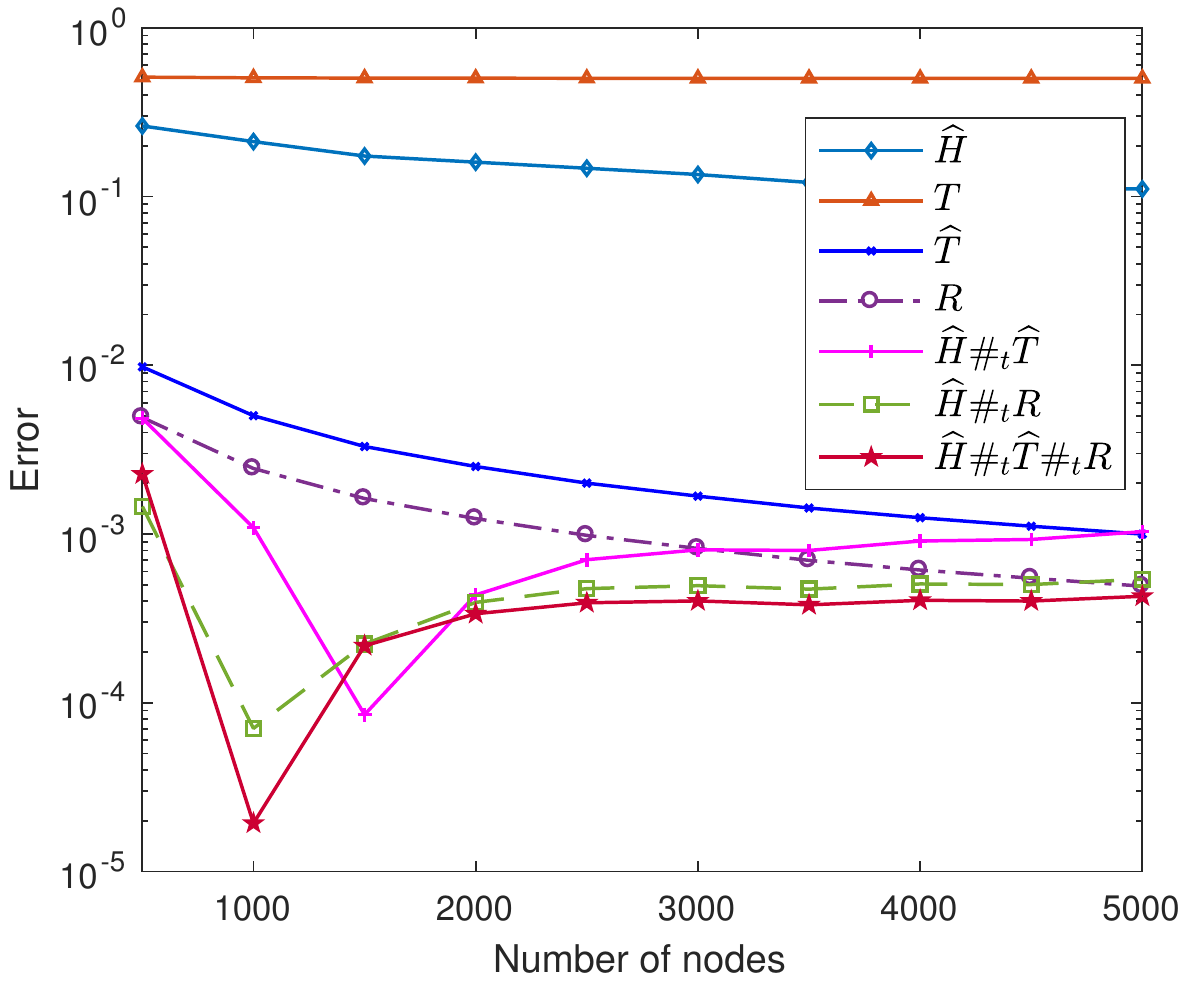}}
  \centerline{(b) ER nodes}\medskip
\end{minipage}

\begin{minipage}[b]{.45\linewidth}
  \centering
  \centerline{\includegraphics[width=1\linewidth]{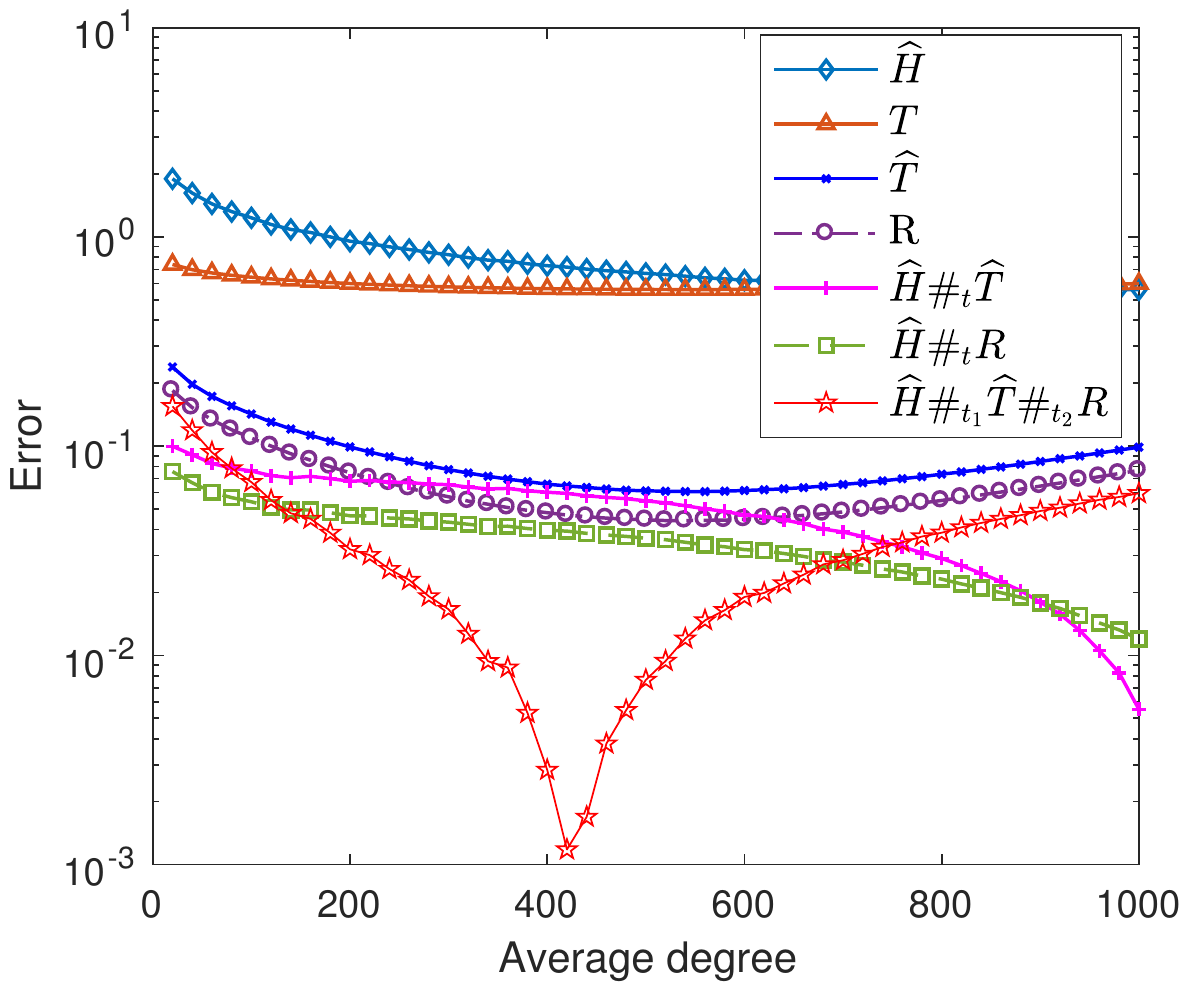}}
  \centerline{(c) BA average degree}\medskip
\end{minipage}
\begin{minipage}[b]{.45\linewidth}
  \centering
  \centerline{\includegraphics[width=1\linewidth]{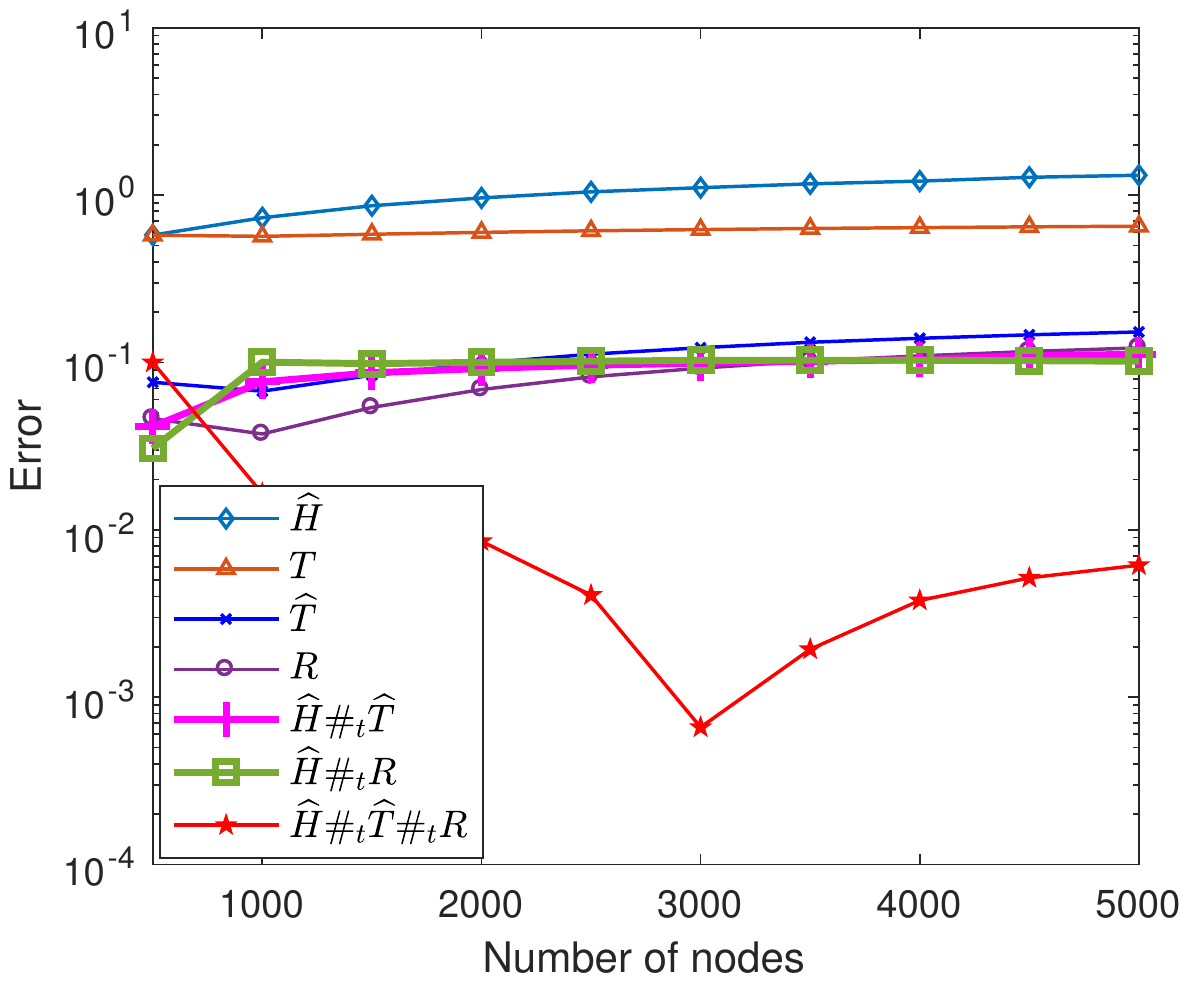}}
  \centerline{(d) BA nodes}\medskip
\end{minipage}

\begin{minipage}[b]{.45\linewidth}
  \centering
  \centerline{\includegraphics[width=1\linewidth]{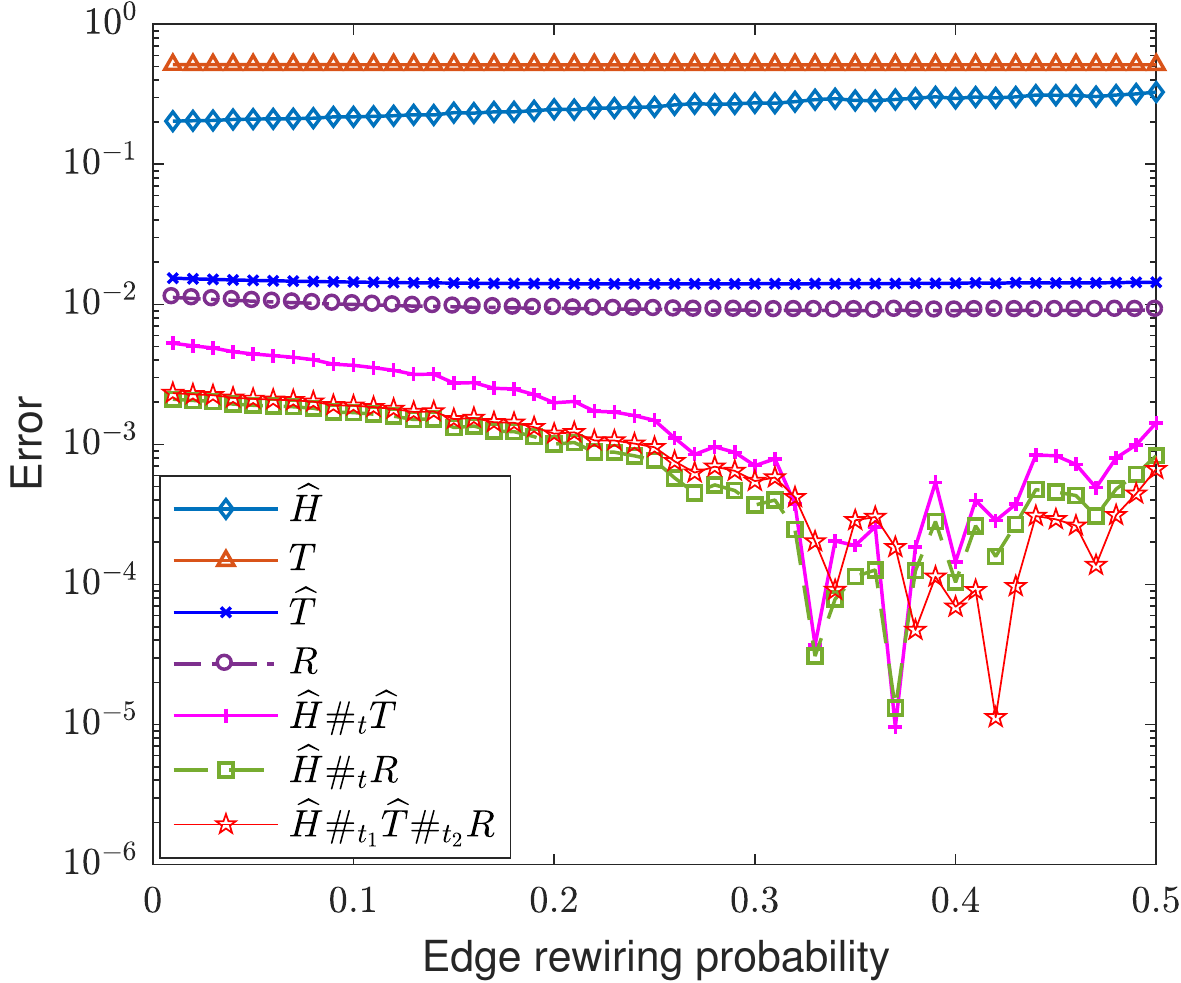}}
  \centerline{(e) WS average degree}\medskip
\end{minipage}
\begin{minipage}[b]{.45\linewidth}
  \centering
  \centerline{\includegraphics[width=1\linewidth]{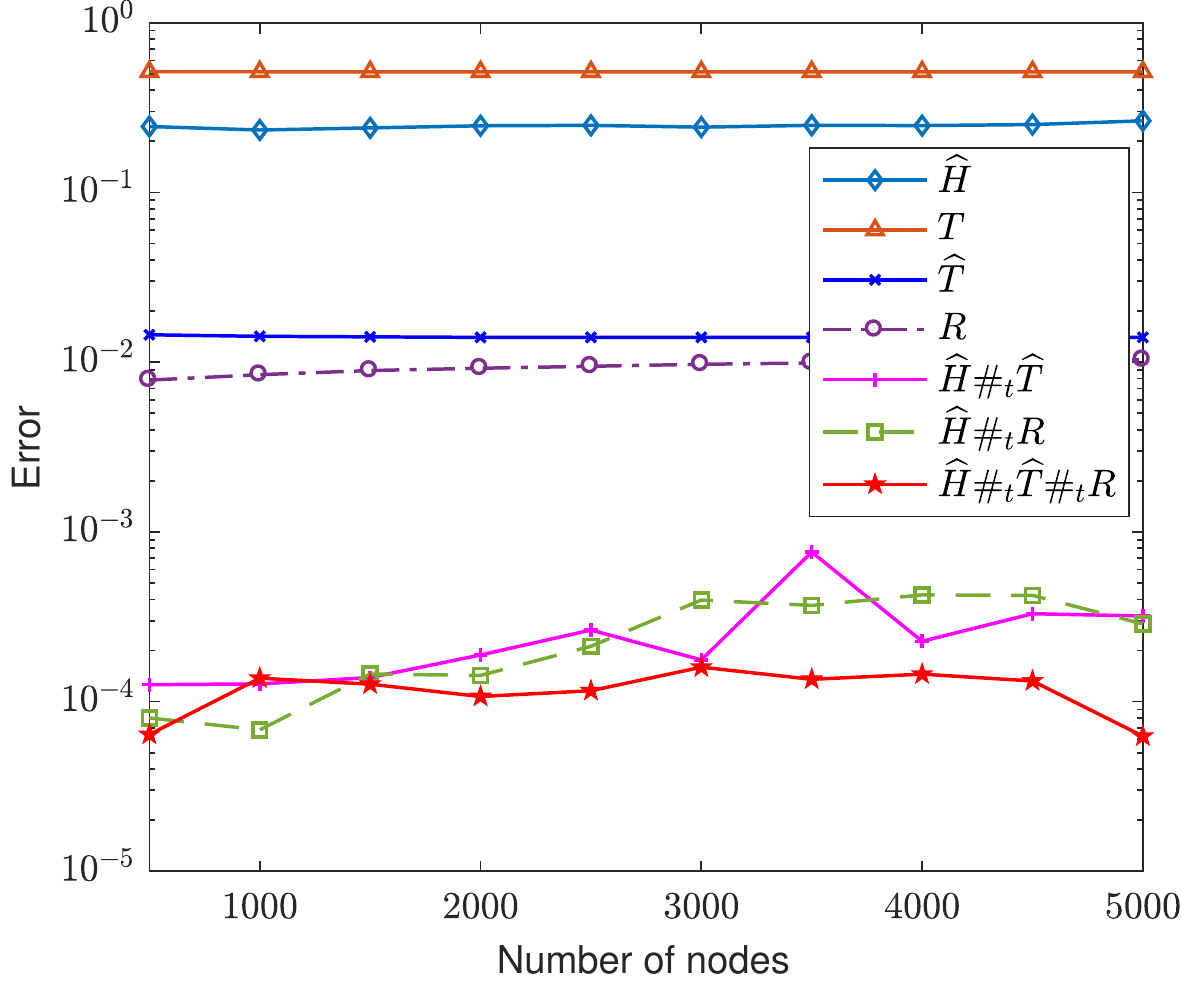}}
  \centerline{(f) WS nodes}\medskip
\end{minipage}
\end{center}
\caption[left]{Error of the approximations for three random models: (i) Erd\H{o}s-R\'{e}nyi (ER) model; (ii) Barab\'{a}si-Albert (BA) model; (iii) Watts-Strogatz (WS) model.}
\label{fig:res1}
\end{figure}

In this section results from various experiments with data sets are provided. All experiments were conducted by MATLAB R2016 on a 16-core machine with 128GB RAM.

Three random graph models are considered: (i) the Erd\H{o}s-R\'{e}nyi (ER) model \cite{ER59,gilbert1959} - the ER model represents two closely related models that were introduced independently and simultaneously. Here we use the ER model which was proposed by Gilbert \cite{gilbert1959}. $G(n,p)$ is denoted as a model with $n$ nodes, and each pair of nodes were linked independently with probability $0 \leq p \leq 1$;
(ii) the Barab\'{a}si-Albert (BA) model \cite{RevModPhys.74.47} - the BA model is a special case of the Price's model. It can generate scale-free graphs in which the degree distribution of the graph follows the power law distribution; 
and (iii) the Watts-Strogatz (WS) model \cite{watts1998collective}  - the WS model generates graphs with small world properties, given a network with N nodes and the mean degree K, initially nodes are linked as a regular ring where each node is connected to $K/2$ nodes in each side, then rewire the edges with probability $0 \leq p \leq 1$.
The approximation error is defined as $| Exact - Approximation |.$ 
The results are averaged over $50$ random trials. 

The simulations demonstrate that Improved Modified Taylor and Improved Radial Projection have best performances.
However, they are required to compute the maximum eigenvalue.
Thus, the computational cost is slightly higher than Radial Projection which does not need to compute the maximum eigenvalue.
Thus considering the time cost, Radial Projection is the superior method.

Theorem \ref{thm:FINGER_eigenmax} (2) states that FINGER-$\widehat{\mathbf{H}}$ is always smaller than the exact von Neumann entropy. On the other hand, by Theorem 3 Modified Taylor-$\widehat{\mathbf{T}}$ is always greater than than the exact von Neumann entropy. 
Additionally, Theorem 3 shows that FINGER $\widehat{\mathbf{H}}(\rho)$ is bounded above by $(1-\lambda_{\max}) \mathbf{H}(\rho)$.
That is, the FINGER is always under-estimated whose error is bigger than $\lambda_{\max}  \mathbf{H}(\rho)$.
In a similar way, $\widehat{\mathbf{T}}$ is always over-estimated. Although we could not find the minimum error mathematically, the simulation results for lots of dataset show that the minimum error is strictly bigger than $0$.

\begin{figure}[t!]
  \centering
    \includegraphics[width=1\textwidth]{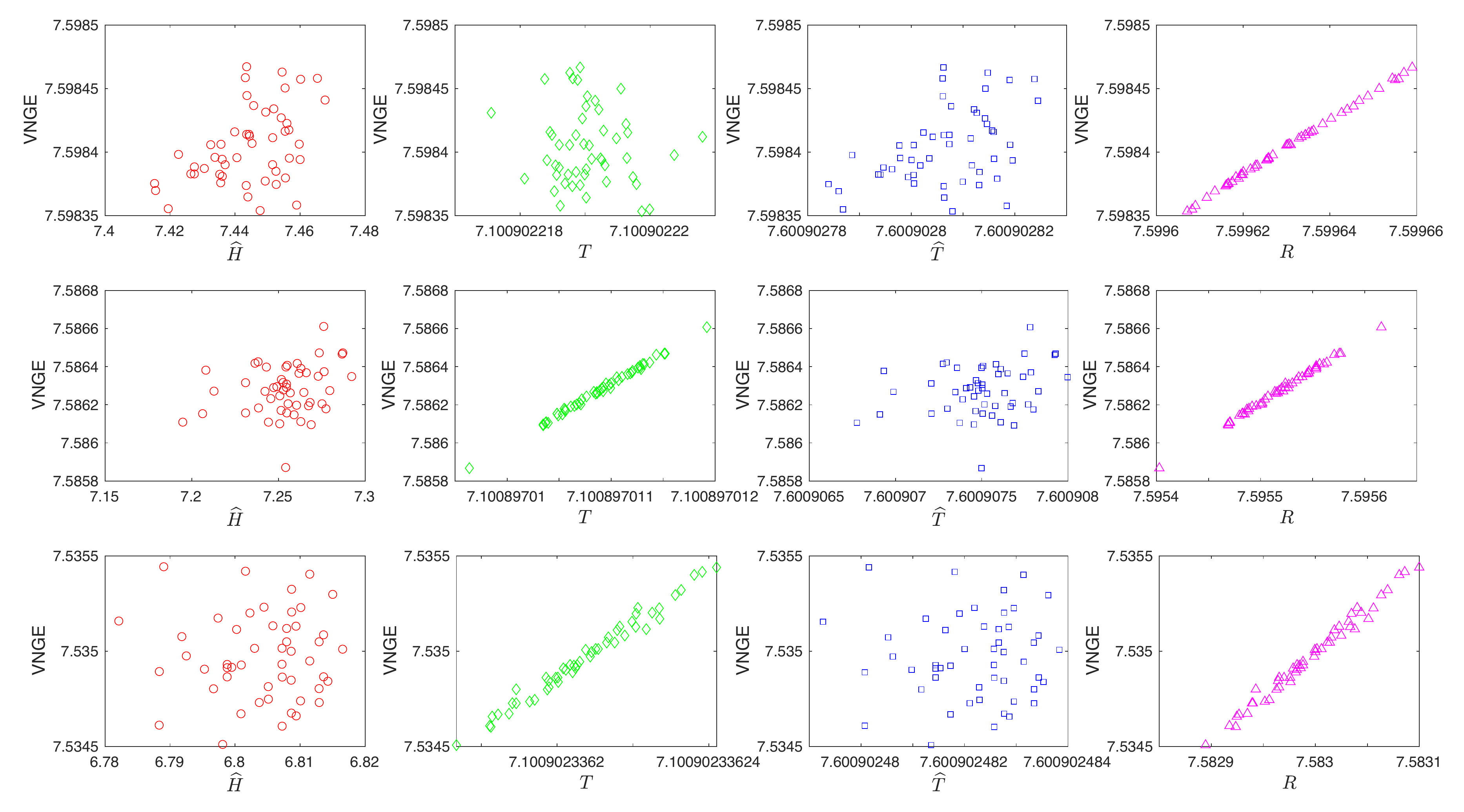}
     \caption{Entropies correlations on the ER model (top), BA model (middle) and WS model (bottom). } \label{fig:correlation}
\end{figure}

In order to further analyze the performance of our proposed algorithms, we perform a correlation study between the exact von Neumann entropy and corresponding approximate value returned by four approximation algorithms on three different random graph model. More specifically, for each model, we generate $50$ graphs. Fig. \ref{fig:correlation} shows the results of the correlation analysis. It is shown that there exists a strong correlation between the exact VNGE and its approximation obtained by Radial Projection on all random graphs. We then observe that correlation between exact VNGE and its approximation got by Taylor is strong on BA and WS model, but becomes much weaker on ER model.
The reader is referred to \cite{2018arXiv180907533M} for a recent analysis of FINGER approximation for graph Laplacians.

\subsection{Real-world datasets}
The real-world datasets in various fields are considered \cite{nr,OPSAHL2009155,OPSAHL2013159}.
We use 137 different number of unweighted networks and 48 different number of weighted networks in different fields for simulations.
The detailed information about datasets on Fig.~\ref{exact_approx1} and Fig.~\ref{exact_approx2} can be found at 
\url{https://github.com/Hang14/RDJ}.
Fig.~\ref{exact_approx1} and Fig.~\ref{exact_approx2} show the scatter points of the von Neumann entropy (y-axis) versus the quadratic approximations (x-axis) for both the unweighted and weighted real-world datasets.
It demonstrates that Modified Taylor and Radial Projection have better performances than FINGER. 
Mixed Quadratic approximation shows the best performance.

\subsection{Time comparison}
Recall that computing von Neumann graph entropy requires $\mathcal{O}(n^3)$ computational complexity. In order to accelerate its computation,
we use the quadratic approximations for the function $f(x)=-x \ln x$. 
Then each approximation can be computed by the purity of the density matrix for a given graph.
Lemma \ref{lemma:quadratic_trace1} shows that computing the purity requires $\mathcal{O}(n+m)$ computational complexity, where $|V| = n$ and $|E| = m$. However, FINGER and Modified Taylor additionally need to compute the maximum eigenvalue whose time complexity is $\mathcal{O}(n^2)$ \cite{DBLP:books/daglib/0086372}. On the other hand, Radial projection shows the best performance with no maximum eigenvalue.
When the original graph is a complete graph then $|E| = n^2$ which is also the upper bound of $|E|$. However, in real life scenario, complete graph is rare, instead, sparse graph are more commonly seen, therefore, the time complexity can remain in a rather linear form.

\begin{figure}[t!]
  \centering
    \begin{subfigure}[b]{0.95\textwidth}
      \centering
        \includegraphics[width=1\linewidth]{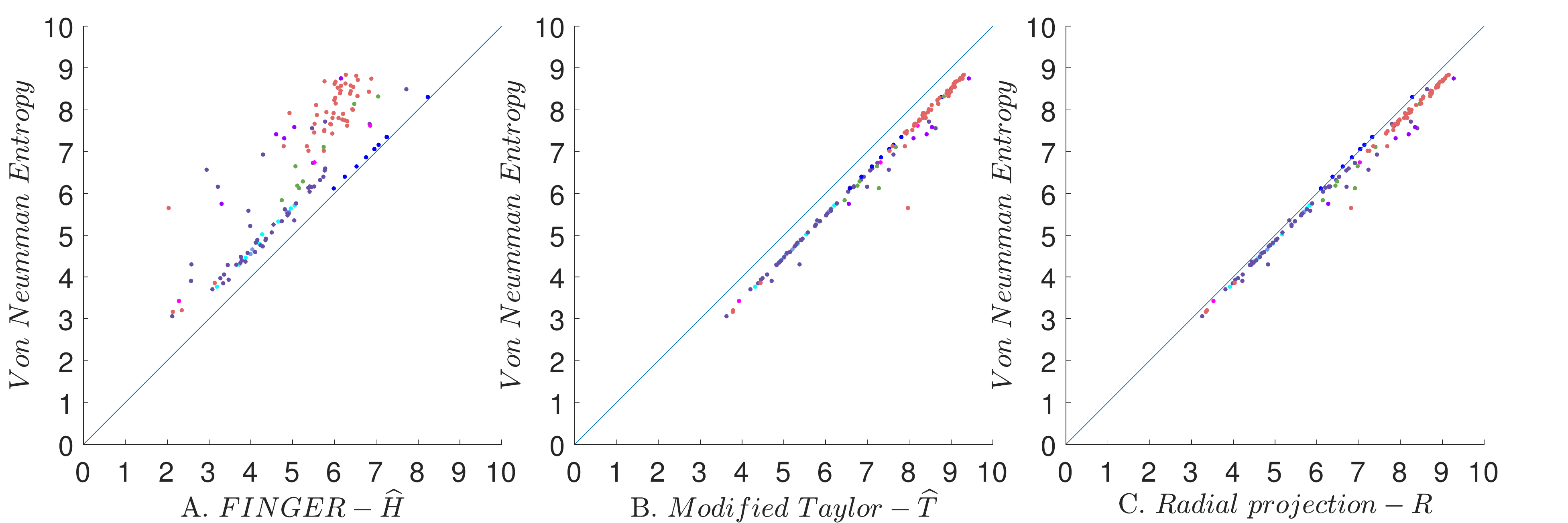}
              \end{subfigure}              
    \begin{subfigure}[b]{0.95\textwidth}
      \centering
        \includegraphics[width=1\linewidth]{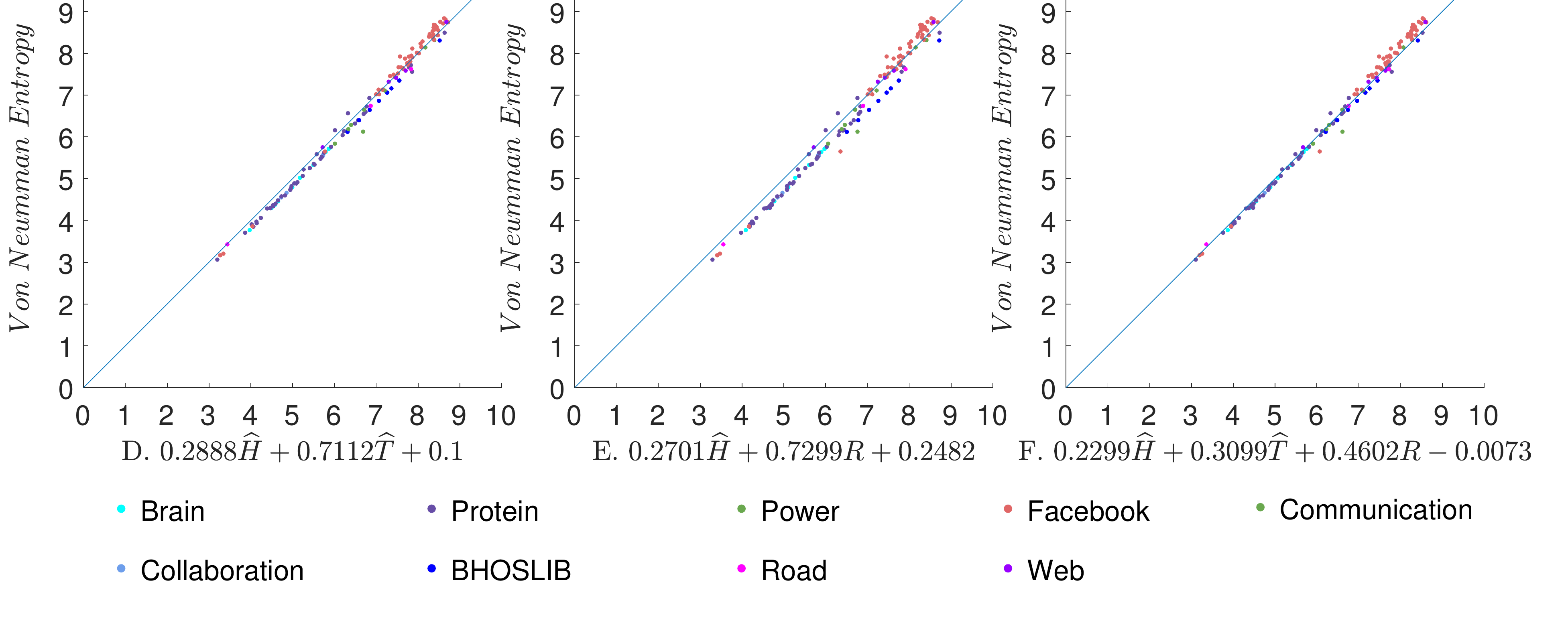}
    \end{subfigure}
    \caption{Exact entropy versus approximate entropy for the unweighted real datasets.}
      \label{exact_approx1}
\end{figure}

\begin{figure}[htp!]
  \centering
    \begin{subfigure}[b]{0.95\textwidth}
      \centering
        \includegraphics[width=1\linewidth]{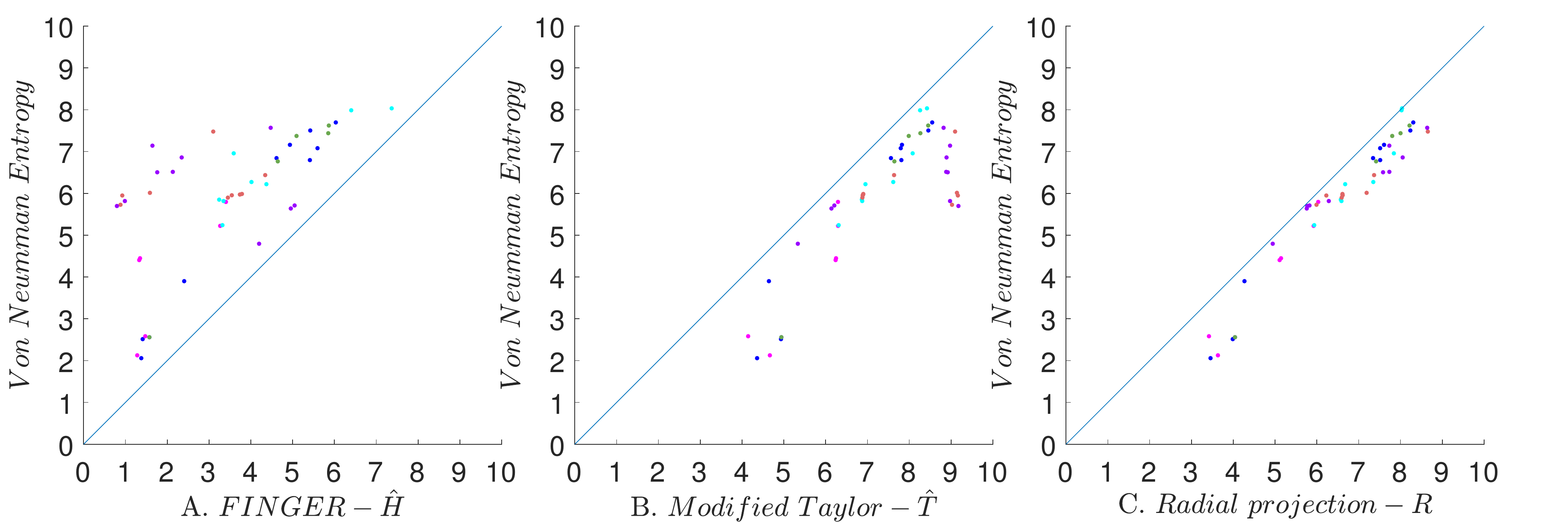}
              \end{subfigure}              
    \begin{subfigure}[b]{0.95\textwidth}
      \centering
        \includegraphics[width=1\linewidth]{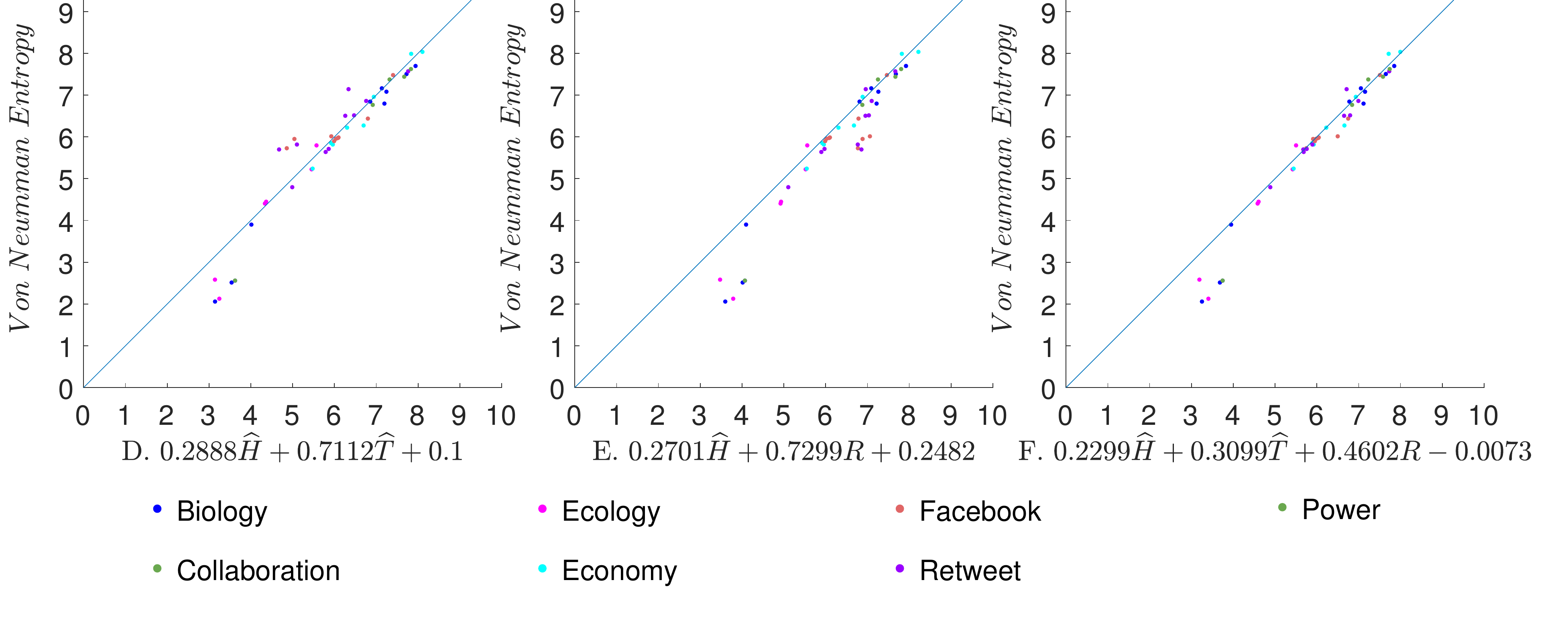}
    \end{subfigure}
    \caption{Exact entropy versus approximate entropy for the weighted real datasets.}
      \label{exact_approx2}
\end{figure}

\section{Applications}\label{sec:applicatoins}
\label{sec:applications}
One major application of von Neumann graph entropy is the computation of Jensen-Shannon distance(JSdist) between any two graphs from a graph sequence \cite{articleehnail}. Given a graph sequence $\mathcal{G}$, the Jensen-Shannon distance of any two graphs $G = (V,E,\bm{W}) \in \mathcal{G}$ and $G' = (V,E',\bm{W}') \in \mathcal{G}$ is defined as
$$\text{JSdist}(G,G') = \sqrt{\NE(\overline{G})-\frac12 [\NE(G)+\NE(G')]},$$ 
where $\overline{G} = (V,\overline{E},\overline{\bm{W}}) = \frac{G \oplus G'}{2}$ is the averaged graph of $G$ and $G'$ such that $\overline{\bm{W}} = \frac{\bm{W}+\bm{W}'}{2}$.
The Jensen-Shannon distance has been proved to be a valid distance metric in \cite{DBLP:journals/tit/EndresS03,PhysRevA.79.052311}.

The Jensen-Shannon distance have been applied into many fields including network analysis \cite{articleehnail} and machine learning \cite{7424294}. Especially, it can be used in anomaly detection and bifurcation detection \cite{8461400}.
\cite{2018arXiv180511769C} demonstrated the validation of using FINGER for computing VNGE. Comparing to the state-of-art graph similarity methods, FINGER yields superior and robust performance for anomaly detection in evolving Wikipedia networks and router communication networks, as well as bifurcation analysis in dynamic genomic networks.
Note that the simulations show that our proposed methods show better performance than FINGER.




\section{Final remarks}

We proposed quadratic approximations for efficiently estimating the von Neumann entropy of large-scale graphs. It reduces the computation of VNGE from cubic complexity to linear complexity for a given graph. 
We finally close with some open problems arisen during our study. 

\begin{itemize}
\item[(1)]
In Theorem \ref{thm:Modified_Taylor_error},  the modified Taylor-$\MT$ is always bigger than VNGE. However, It is still open if there exists any error bound for $\MT$ which is same as or similar to Theorem \ref{thm:FINGER_eigenmax} (3).

\item[(2)]
The simulations show that Radial projection-$\RP$ have superior performances without any information about the maximum eigenvalue. However, it is questionable if there exists some error bound between von Neumann entropy and Radial projection.

\item[(3)] 
Fig. \ref{fig:res1} shows that the error for weighted means have interesting behaviors. Further analysis is needed with different class of graphs.
It remains for future work.

\end{itemize}

\section*{Acknowledgement}
This work of H. Choi and Y. Shi was partially supported  by Shanghai Sailing Program under Grant 16YF1407700, and National Nature Science Foundation of China (NSFC) under Grant No. 61601290.


\section*{References}
\bibliographystyle{plain}
\bibliography{refs.bib}

\end{document}